\documentclass{stacs}

\usepackage{pstricks}
\usepackage{pst-coil}
\usepackage{amssymb, amsmath, amsthm}
\usepackage{newalg}
\usepackage[square,comma,numbers]{natbib}
\usepackage{floatflt}
\usepackage{wrapfig}
\usepackage{ifthen}

\newboolean{journal}
\setboolean{journal}{true}

\newcommand{\benefit}[2]{\Delta(#1,#2)}

\newcommand{\OPT}{\mathrm{OPT}}
\newcommand{\DL}{\mathrm{DL}}
\newcommand{\IP}{\mathrm{IP}}
\newcommand{\LP}{\mathrm{LP}}
\newcommand{\subtree}[1]{T_{#1}}

\newcommand{\suchthat}{\, | \, }

\newcommand{\pruned}[1]{\widehat{#1}}

\newcommand{\fancybox}[1]
{ \begin{center}\noindent\psframebox*[framesep=7pt]{\psframebox[boxsep=true,framearc=0.05,framesep=9pt]{#1}} \end{center} }

\newcommand{\collection}{\mathcal}
\newcommand{\myalgo}{\mathcal}

\newcommand{\vf}[1]{#1}
\newcommand{\mf}[1]{#1}
\newcommand{\vXv}[2]{\vf{#1} \cdot \vf{#2}}
\newcommand{\mXv}[2]{\mf{#1} \vf{#2}}

\begin{document}

\ifthenelse{\not \boolean{journal}}
{\title[Lagrangian Relaxation and Partial Cover]{Lagrangian Relaxation and Partial Cover \\ (Extended Abstract)}}
{\title[Lagrangian Relaxation and Partial Cover]{Lagrangian Relaxation and Partial Cover }}

\author[lab1]{J. Mestre}{Juli\'{a}n Mestre}
\address[lab1]{Max-Planck-Institute f\"{u}r Informatik, Saarbr\"{u}cken, Germany.}
\email{jmestre@mpi-inf.mpg.de}  

\thanks{Supported by NSF Award CCF 0430650, a University of Maryland Dean's Dissertation Fellowship, and partly by an Alexander von Humboldt Fellowship.  }



\begin{abstract}

Lagrangian relaxation has been used extensively in the design of approximation algorithms. This paper studies its strengths and limitations when applied to Partial Cover.

We show that for Partial Cover in general no algorithm that uses Lagrangian relaxation and a Lagrangian Multiplier Preserving (LMP) $\alpha$-approximation as a black box can yield an approximation factor better than~$\frac{4}{3} \alpha$. This matches the upper bound given by K\"{o}nemann \textit{et al.} (\emph{ESA 2006}, pages 468--479).

Faced with this limitation we study a specific, yet broad class of covering problems: Partial Totally Balanced Cover. By carefully analyzing the inner workings of the LMP algorithm we are able to give an almost tight characterization of the integrality gap of the standard linear relaxation of the problem. As a consequence we obtain improved approximations for the Partial version of Multicut and Path Hitting on Trees, Rectangle Stabbing, and Set Cover with $\rho$-Blocks.
\end{abstract}

\maketitle


\section{Introduction}

\ifthenelse{\boolean{journal}}
{Lagrangian relaxation has been used extensively in the design of approximation algorithms for a variety of problems such as TSP \cite{HK70,HK71}, $k$-MST \cite{G96,AK00,CRW01,G05}, partial vertex cover \cite{H98}, $k$-median \cite{JV01,CG99,ARS03}, MST with degree constraints \cite{KR02} and budgeted MST \cite{RG96}.}{Lagrangian relaxation has been used extensively in the design of approximation algorithms for a variety of problems such as $k$-MST \cite{G96,CRW01,G05}, $k$-median \cite{JV01,CG99}, MST with degree constraints \cite{KR02} and budgeted MST \cite{RG96}.}

 In this paper we study the strengths and limitations of Lagrangian relaxation applied to the Partial Cover problem. Let $\collection{S}$ be collection of subsets of a universal set $U$ with cost $c:\collection{S}\rightarrow R_+$ and profit $p:U \rightarrow R_+$, and let $P$ be a target coverage parameter. A set $\collection{C} \subseteq \collection{S}$ is a \emph{partial cover} if the overall profit of elements covered by $\collection{C}$ is at least $P$. The objective is to find a minimum cost partial cover.

The high level idea behind Lagrangian relaxation is as follows. In an IP formulation for Partial Cover, the constraint enforcing that at least $P$ profit is covered is \emph{relaxed}: The constraint is multiplied by a parameter $\lambda$ and lifted to the objective function. This relaxed IP corresponds, up to a constant factor, to the prize-collecting version of the underlying covering problem in which there is no requirement on how much profit to cover but a penalty of $\lambda\, p(i)$ must be paid if we leave element $i \in U$ uncovered. An approximation algorithm for the prize-collecting version having the Lagrangian Multiplier Preserving (LMP) property\footnote{The definition of the LMP property is outlined in Section 2.} is used to obtain values $\lambda_1$ and $\lambda_2$ that are close together for which the algorithm produces solutions $\collection{C}_1$ and $\collection{C}_2$ respectively. These solutions are such that $\collection{C}_1$ is inexpensive but unfeasible (covering less than $P$ profit), and $\collection{C}_2$ is feasible (covering at least $P$ profit) but potentially very expensive. Finally, these two solutions are combined to obtain a cover that is both inexpensive and feasible.

Broadly speaking there are two ways to combine $\collection{C}_1$ and $\collection{C}_2$. One option is to treat the approximation algorithm for the prize-collecting version as a black box, only making use of the LMP property in the analysis. Another option is to focus on a particular LMP algorithm and exploit additional structure that it may offer. Not surprisingly, the latter approach has yielded better approximation guarantees. For example, for $k$-median compare the 6-approximation of \citet{JV01} to the 4-approximation of \citet{CG99}; for $k$-MST compare the 5-factor to the 3-factor approximation due to \citet{G96}.

The results in this paper support the common belief regarding the inherent weakness of the black-box approach. First, we show a lower bound on the approximation factor achievable for Partial Cover in general using Lagrangian relaxation and the black-box approach that matches the recent upper bound of \citet{KPS06}. To overcome this obstacle, we concentrate on Kolen's algorithm for Prize-Collecting Totally Balanced Cover \cite{K82}. By carefully analyzing the algorithm's inner workings we identify structural similarities between $\collection{C}_1$ and $\collection{C}_2$, which we later exploit when combining the two solutions. As a result we derive an almost tight characterization of the integrality gap of the standard linear relaxation for Partial Totally Balanced Cover. This in turn implies improved approximation algorithms for a number of related problems.

\subsection{Related Work}

Much work has been done on covering problems because of both their simple and elegant formulation, and their pervasiveness in different application areas. In its most general form the problem, also known as Set Cover, cannot be approximated within $(1 -\epsilon) \ln |U|$ unless $NP \subseteq \mathrm{DTIME}(|U|^{\log \log |U|})$~\cite{F98}. Due to this hardness, easier, special cases have been studied. 

A general class of covering problems that can be solved efficiently are those whose element-set incidence matrix is balanced. A $0,1$ matrix is \emph{balanced} if it does not contain a square submatrix of odd order with row and column sums equal to 2. These matrices were introduced by Berge \cite{B72} who showed that if $A$ is balanced then the polyhedron $\{ \vf{x}\!\geq\! \vf{0} :  \mXv{A}{x} \!\geq \!\vf{1} \}$ is integral. A $0,1$ matrix is \emph{totally balanced} if it does not contain a square submatrix with row and column sums equal to 2 and no identical columns. Kolen \cite{K82} gave a simple primal-dual algorithm that solves optimally the covering problem defined by a totally balanced matrix. A $0,\pm 1$ matrix is \emph{totally unimodular} if every square submatrix has determinant 0 or $\pm 1$. Although totally balanced and totally unimodular matrices are subclasses of balanced matrices, the two classes are neither disjoint nor one is included in the other. 

Beyond this point, even minor generalizations can make the covering problem hard. For example, consider the \emph{vertex cover} problem: Given a graph $G=(V,E)$ we are to choose a minimum size subset of vertices such that every edge is incident on at least one of the chosen vertices. If $G$ is bipartite, the element-set incidence matrix for the problem is totally unimodular; however, if $G$ is a general graph the problem becomes NP-hard \cite{Karp72}. Numerous approximation algorithms have been developed for vertex cover~\cite{Hoc-book}. The best known approximation factor for general graphs is \mbox{$2 - o(1)$} \cite{BYE85,H02,Karakostas05}; yet, after 25 years of study, the best constant factor approximation for vertex cover remains 2 \cite{C83,BYE81,H82}. This lack of progress  has led researchers to seek generalizations of vertex cover that can still be approximated within twice of optimum. One such generalization is the \emph{multicut} problem on trees: Given a tree $T$ and a collection of pairs of vertices, a cover is formed by a set of edges whose removal separates all pairs. The problem was first studied by \citet{GVY97} who gave an elegant primal-dual 2-approximation. 

A notable shortcoming of the standard set cover formulation is that certain hard-to-cover elements, also known as \emph{outliers} \cite{CKMN01}, can render the optimal solution very expensive. 
Motivated by the presence of outliers, the unit-profit partial version calls for a collection of sets covering not all but a specified number $k$ of elements. Partial Multicut, a.k.a. $k$-Multicut, was recently studied independently by \citet{LS05} and by \citet{GNS06}, who gave a $\frac{8}{3} + \epsilon$ approximation algorithm. This scheme was generalized by \citet{KPS06} who showed how to design a $\frac{4}{3} \alpha + \epsilon$ approximation for any covering problem using Lagrangian relaxation and an $\alpha$-LMP approximation as a black box. (Their algorithm runs in time polynomial on $|U|, |\collection{S}|^{\frac{1}{\epsilon}}$ and the running time of the $\alpha$-LMP approximation.)

\subsection{Our Results and Outline of the Paper}

Section~\ref{section:lowerbound} shows that for Partial Cover in general no algorithm that uses Lagrangian relaxation and an $\alpha$-LMP approximation as a black box can yield an approximation factor better than $\frac{4}{3} \alpha$.
In Section~\ref{section:P-TBC} we give an almost tight characterization of the integrality gap of the standard LP for Partial Totally Balanced Cover, settling a question posed by \citet{GNS06}. Our approach is based on Lagrangian relaxation and Kolen's algorithm. We prove that $\IP \leq \left(1+\frac{1}{3^{k-1}}\right) \LP + k\, c_{\max}$ for any $k\! \geq\! 1$, where $\IP$ and $\LP$ are the costs of the optimal integral and fractional solutions respectively and $c_{\max}$ is the cost of the most expensive set in the instance. The trade-off between additive and multiplicative error is not an artifact of our analysis or a shortcoming of our approach. On the contrary, this is precisely how the integrality gap behaves. More specifically, we show a family of instances where $\IP > \left(1+ \frac{1}{3^{k-1}}\right) \LP + \frac{k}{2} c_{\max}$. In other words, there is an unbounded additive gap in terms of $c_{\max}$ but as it grows the multiplicative gap narrows exponentially fast.

Finally,
\ifthenelse{\boolean{journal}}{in Section~\ref{sec:applications}}{} we show how the above result can be applied, borrowing ideas \mbox{from \cite{GIK02,HS05a,GNS06}}, to get a $\rho + \epsilon$ approximation or a quasi-polynomial time $\rho$-approximation for covering problems that can be expressed with a suitable combination of $\rho$ totally-balanced matrices. This translates into improved approximations for a number of problems: a $2 + \epsilon$ approximation for the Partial Multicut on Trees \cite{LS05,GNS06}, a $4 + \epsilon$  approximation for Partial Path Hitting on Trees \cite{PS06}, a 2-approximation for Partial Rectangle Stabbing \cite{GIK02}, and a $\rho$ approximation for Partial Set-Cover with $\rho$-blocks \cite{HS05a}. In addition, the $\epsilon$ can be removed from the first two approximation guarantees if we allow quasi-polynomial time.
It is worth noting that prior to this work, the best approximation ratio for all these problems could be achieved with the framework of \citet{KPS06}. In each case our results improve the approximation ratio by a $\frac{4}{3}$ multiplicative factor. 
\ifthenelse{\not \boolean{journal}}{Due to lack of space these results only appear in the full version\footnote{Full version available at {\tt http://arxiv.org/abs/0712.3936}} of the paper.}{}

\section{Lagrangian relaxation} \label{section:LR}

Let $\collection{S} = \{1, \ldots, m\}$ be a collection of subsets of a universal set $U =\{1, \ldots, n\}$. Each set has a cost specified by $\vf{c} \in R_+^m$, and each element has a profit specified by $\vf{p} \in R_+^n$. Given a target coverage $P$, the objective of the Partial Cover problem is to find a minimum cost solution $\collection{C} \subseteq \collection{S}$ such that $p(\collection{C}) \geq P$, where the notation $p(\collection{C})$ denotes the overall profit of elements covered by $\collection{C}$. The problem is captured by the IP below. Matrix $A = \{a_{ij}\}\in \{0,1\}^{n \times m} $ is an element-set incidence matrix, that is, $a_{ij}=1$ if and only if element $i \in U$ belongs to set $j \in \collection{S}$; variable $x_j$ indicates whether set $j$ is chosen in the solution $\collection{C}$; variable $r_i$ indicates whether element $i$ is left uncovered.

Lagrangian relaxation is used to get rid of the constraint bounding the profit of uncovered elements to be at most $p(U) - P$. The constraint is multiplied by the parameter $\lambda$, called Lagrange Multiplier, and is lifted to the objective function. The resulting IP corresponds, up to the constant $\lambda\, (p(U)-P)$ factor in the objective function, to the prize-collecting version of the covering problem, where the penalty for leaving element $i$ uncovered is $\lambda p_i$.
\begin{center}
\vspace{-0.5em}
\begin{minipage}[t]{5cm}
\begin{gather*}
\begin{array}{r@{\hspace{.7ex}}l}
\multicolumn{2}{c}{\displaystyle \min \ \vXv{c}{x}} \\[0.4cm]
\mXv{A}{x} + \mXv{I}{r} \geq & \vf{1} \hspace{0.5cm}  \\[0.2cm]
\vXv{p}{r} \leq & p(U)-P \\[0.2cm]
r_i, x_j \in & \{ 0, 1\} \\[0.4cm]
\end{array}
\end{gather*}
\end{minipage}
\begin{pspicture}(0,1.75)(2.5,1.75)
\psline[linewidth=1.25pt,arrowinset=0,arrowlength=1]{->}(2.5,0)
\rput(1.25,-.4){\parbox{2cm}{\center \footnotesize Lagrangian\\ Relaxation}}
\end{pspicture}
\begin{minipage}[t]{7cm}
\begin{gather*}
\begin{array}{r@{\hspace{.7ex}}l}
\multicolumn{2}{c}{\displaystyle \min\ \vXv{c}{x} +  \lambda \vXv{p}{r} - \lambda\, (p(U)-P)} \\[0.4cm]
\hspace{4em}\mXv{A}{x} + \mXv{I}{r} \geq & \vf{1} \\[0.2cm]
r_i, x_j \in & \{0,1\}
\end{array}
\end{gather*}
\end{minipage}
\end{center}

Let $\OPT$ be the cost of an optimal partial cover and $\mbox{OPT-PC}(\lambda)$ be the cost of an optimal prize-collecting cover for a given $\lambda$. Let $\myalgo{A}$ be an $\alpha$-approximation for the prize-collecting variant of the problem. Algorithm $\myalgo{A}$ is said to have the Lagrangian Multiplier Preserving (LMP) property if it produces a solution $\collection{C}$ such that
\begin{equation}
c(\collection{C}) + \alpha \,\lambda  \big(p(U)-p(\collection{C})\big)  \leq \alpha\, \mbox{OPT-PC}(\lambda) \label{eqn:LMP}.
\end{equation}

Note that $\mbox{OPT-PC}(\lambda) \leq \OPT + \lambda\, (p(U)- P)$. Thus, 
\begin{equation}
c(\collection{C})  \leq \alpha \Big( \OPT +  \lambda\, \big(p(\collection{C}) - P\big) \Big). \label{eqn:cost-C}
\end{equation}

Therefore, if we could find a value of $\lambda$ such that $\collection{C}$ covers exactly $P$ profit then $\collection{C}$ is $\alpha$-approximate. However, if $p(\collection{C})<P$, the solution is not feasible, and if $p(\collection{C})>P$, equation \eqref{eqn:cost-C} does not offer any guarantee on the cost of $\collection{C}$. Unfortunately, there are cases where no value of $\lambda$ produces a solution covering exactly $P$ profit. Thus, the idea is to use binary search to find two values $\lambda_1$ and $\lambda_2$ that are close together and are such that $\myalgo{A}(\lambda_1)$ covers less, and $\myalgo{A}(\lambda_2)$ covers more than $P$ profit. The two solutions are then combined in some fashion to produce a feasible cover.

\section{Limitations of the black-box approach} \label{section:lowerbound}

A common way to combine the two solutions returned by the $\alpha$-LMP is to treat the algorithm as a black box, solely relaying on the LMP property~\eqref{eqn:LMP} in the analysis. More formally, an algorithm for Partial Cover that uses Lagrangian relaxation and an $\alpha$-LMP approximation $\myalgo{A}$ as a black box is as follows. First, we are allowed to run $\myalgo{A}$ with as many different values of $\lambda$ as desired; then, the solutions thus found are combined to produce a feasible partial cover. No computational restriction is placed on the second step, except that only sets returned by $\myalgo{A}$ may be used.

\begin{theorem} \label{theorem:lowerbound} In general, the Partial Cover problem cannot be approximated better than $\frac{4}{3} \alpha$ using Lagrangian relaxation and an $\alpha$-LMP algorithm $\myalgo{A}$ as a black box.
\end{theorem}

\def\horizontalSet#1{
  \psframe[framearc=.5](-1,0)(14,2)
  \rput(2,1){\cluster}
  \rput(5,1){\cluster}
  \rput(11,1){\cluster}
  \psdot[dotsize=2pt](13,1)
  \psdot[dotsize=2pt](0,1)
  \rput(-2.5,1){$B_{#1}$}
}

\def\verticalSet#1
{
   \psframe[framearc=0.5](0,0)(2,13)
   \rput(1,14.5){$A_{#1}$}
}

\def\cluster{ \\
\psset{dotsize=1pt}\\
\psset{unit=2pt}\\ 
\multirput(-2,0)(1,0){5}{\psdots(0,1)(0,0)(0,-1)}\\
\multirput(-1,0)(1,0){3}{\psdots(0,2)(0,-2)}
}

\begin{floatingfigure}[r]{5cm}
\psset{unit=7pt}
\begin{pspicture}(-3,-3)(14,15)
\rput(0,9){\horizontalSet{1}}
\rput(0,6){\horizontalSet{2}}
\rput(0,0){\horizontalSet{q}}
\rput(1,-1){\verticalSet{1} }
\rput(4,-1){\verticalSet{2} }
\rput(10,-1){\verticalSet{q} }
\rput(8,13.5){$\ldots$}
\rput{90}(-2.6,4){$\ldots$}
\end{pspicture}

\end{floatingfigure}

Let $A_1,\ldots A_q$ and $B_1, \ldots B_q$ be sets as depicted on the right. For each $i$ and $j$ the intersection $A_i \cap B_j$ consists of a cluster of $q$ elements. There are $q^2$ clusters. Set $A_i$ is made up of $q$ clusters; set $B_i$ is made up of $q$ clusters and two additional elements (the leftmost and rightmost elements in the picture.) Thus $|A_i| = q^2$ and $|B_i| = q^2 + 2$. 
In addition, there are sets $O_1, \ldots, O_q$, which are not shown in the picture. Set $O_i$ contains one element from each cluster and the leftmost element of $B_i$. Thus $|O_i| = q^2 + 1$. 
The cost of $O_i$ is $\frac{1}{q}$, the cost of $A_i$ is $\frac{2\,\alpha}{3\,q}$, and the cost of $B_i$ is $\frac{4\,\alpha}{3\,q}$.
Every element has unit profit and the target coverage is $P = q^3 + q$. 
It is not hard to see that $O_1, \ldots, O_q$ is an optimal partial cover with a cost of 1. \ifthenelse{\boolean{journal}}
{Furthermore, for any value of $\lambda$ the optimal prize-collecting cover uses sets of one kind.

\begin{lemma} \label{lemma:structure-prize-collecting} For all values of $\lambda$, the optimal prize-collecting cover in our instance is either the empty cover or $A_1,\ldots,A_q$ or $B_1, \ldots, B_q$ or $O_1, \ldots, O_q$.
\end{lemma}

\begin{proof}
Suppose that some subsets of the $B$-sets and the $O$-sets have already been chosen. Every $A$-set has the same marginal benefit (the penalty of elements not-yet-covered by the set minus its cost) the independent of the other $A$-sets. Thus, in an optimal solution either all the $A$-sets are chosen or none is.

Now suppose that the $O$-sets and the $A$-sets have already been chosen. In this case there are two marginal benefits for the $B$-sets, depending on whether an $O$-set already covers the leftmost element of the $B$-set or not. Thus, an optimal strategy for the $B$-sets is either to choose none, all, or the complement of the $O$-sets, i.e., $B_i$ is chosen if and only if $O_i$ is not chosen. A solution where $B$-sets and $O$-sets complement each other is always worse than either choosing all the $O$-sets and no $B$-sets, or vice versa. Thus, in an optimal solution either all the $B$-sets are chosen or none is. It follows that the same holds for the $O$-sets; i.e., they are either all in or all out.

Notice that if the $B$-sets are chosen then there is no reason to choose any of the remaining sets. If the $B$-sets are not chosen and the $O$-sets are chosen then there is no reason to choose the $A$-sets. The only possibility left is to choose only the $A$-sets, or the empty cover.
\end{proof}

}{}

The $\alpha$-LMP approximation algorithm we use has the unfortunate property that it never returns sets from the optimal solution.

\begin{lemma} \label{lemma:naughty-LMP} There exists an $\alpha$-LMP approximation $\myalgo{A}$ that for the above instance and any value of $\lambda$ outputs either $\emptyset$ or $A_1, \ldots, A_q$ or $B_1, \ldots, B_q$.
\end{lemma}

\ifthenelse{\boolean{journal}}
{
\begin{proof}
Let us denote each of the four alternatives in Lemma~\ref{lemma:structure-prize-collecting} by $E$, $A$, $B$ and $O$. The prize-collecting cost of these covers is $(q^3 + 2q) \lambda$, $\frac{2}{3} \alpha + 2 q\, \lambda$, $\frac{4}{3} \alpha$ and $1 + q\, \lambda$ respectively. For a fixed value of $\lambda$, the minimum of these four quantities corresponds to the cost of the optimal prize-collecting cover. The line below shows which solution is optimal as a function of $\lambda$.

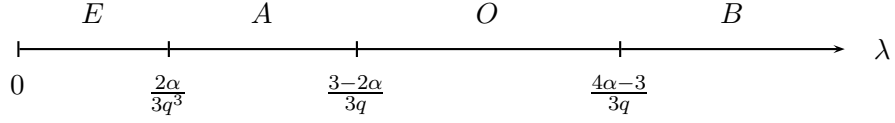
\begin{center}
  \begin{pspicture}(0,-1)(11,1)
    \psset{arrows=|-}
    \psset{labelsep=10pt}
    \uput[d](0,0){$0$}
    \psline(0,0)(2,0)
    \uput[u](1,0){$E$}
    \uput[d](2,0){$\frac{2\alpha}{3q^3}$}
    \psline(2,0)(4.5,0)
    \uput[u](3.25,0){$A$}
    \uput[d](4.5,0){$\frac{3- 2\alpha}{3 q}$}
    \psline(4.5,0)(8,0)
    \uput[u](6.25,0){$O$}
    \uput[d](8,0){$\frac{4\alpha -3 }{3q}$}
    \psline[arrows=|->](8,0)(11,0)
    \uput[u](9.5,0){$B$}
    \rput(11.5,0){$\lambda$}
  \end{pspicture}
\end{center}

If $\alpha$ is big enough then the $A$-interval disappears. Namely, if $\frac{2\alpha}{3q^3} > \frac{3- 2\alpha}{3 q}$ then the line looks as follows

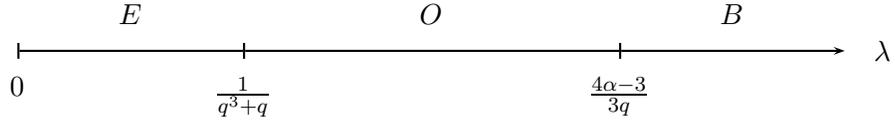
\begin{center}
  \begin{pspicture}(0,-1)(11,1)
    \psset{arrows=|-}
    \psset{labelsep=10pt}
    \uput[d](0,0){$0$}
    \psline(0,0)(3,0)
    \uput[u](1.5,0){$E$}
    \psline(3,0)(8,0)
    \uput[d](3,0){$\frac{1}{q^3 + q}$}
    \uput[u](5.5,0){$O$}
    \uput[d](8,0){$\frac{4\alpha -3 }{3q}$}
    \psline[arrows=|->](8,0)(11,0)
    \uput[u](9.5,0){$B$}
    \rput(11.5,0){$\lambda$}
  \end{pspicture}
\end{center}

We are now ready to describe the $\alpha$-LMP algorithm. Given a value of $\lambda$, we need to decide whether to output the empty cover, the $A$-sets or the $B$-sets. If $\lambda$ falls in the interval corresponding to one of these three solutions then output that solution; since the cover output is optimal, the LMP property~\eqref{eqn:LMP} follows trivially.

If $\lambda$ falls in the $O$-interval and $\lambda \leq \frac{1}{3q}$ then output the $A$-sets; the LMP property holds since 
\[c(A) + \alpha\, \lambda\, \overline{p(A)} = \frac{2}{3} \alpha + \alpha\,  \lambda \,  2 q \leq \alpha (1 + \lambda\, q) = \alpha \Big(c(O) + \lambda\, (p(U) - p(O)) \Big).\]

If $\lambda$ falls in the $O$-interval and $\lambda > \frac{1}{3q}$ then output the $B$-sets; the LMP property holds since
\[c(B) + \alpha\, \lambda\, \overline{p(B)} = \frac{4}{3} \alpha < \alpha (1 + \lambda\, q) = \alpha \Big(c(O) + \lambda\, (p(U) - p(O))\Big).\]

\end{proof}
}{}

\ifthenelse{\not \boolean{journal}}
{The proof that such an algorithm exists is given in the full version of the paper.}{} Hence, if we use $\myalgo{A}$ as a black box we must build a partial cover with the sets $A_1, \ldots, A_q$ and $B_1, \ldots, B_q$. Note that in order to cover $q^2 + q$ elements either all $A$-sets, or all $B$-sets must be used. In the first case $\frac{q}{2}$ additional $B$-sets are needed to attain feasibility, and the solution has cost $\frac{4}{3}\alpha$; in the second case the solution is feasible but again has cost $\frac{4}{3}\alpha$. Theorem~\ref{theorem:lowerbound} follows.

One assumption usually made in the literature  \cite{AK00,GKS04,KPS06} is that $c_{\max} = \max_{j} c_j \leq \epsilon\, \OPT$, for some constant $\epsilon> 0$, or more generally an additive error in terms of $c_{\max}$ is allowed. This does not help in our construction as $c_{\max}$ can be made arbitrarily small by increasing~$q$.

Admittedly, our lower bound example belongs to a specific class of covering problem (every element belongs to at most three sets) and although the example can be embedded into a partial totally unimodular covering problem 
\ifthenelse{\boolean{journal}}{(see Appendix~\ref{appendix:lowerbound-TU})}{(see full version)}, it is not clear how to embed the example into other classes. Nevertheless, the  $\frac{4}{3}\alpha$ upper bound of Koneman et el. \cite{KPS06} makes no assumption about the underlying problem, only using the LMP property \eqref{eqn:LMP} in the analysis. It was entirely conceivable that the $\frac{4}{3} \alpha$ factor could be improved using a different merging strategy---Theorem~\ref{theorem:lowerbound} precludes this possibility.

\section{Partial Totally Balanced Cover} \label{section:P-TBC}

In order to overcome the lower bound of Theorem~\ref{theorem:lowerbound}, one must concentrate on a specific class of covering problems or make additional assumptions about the $\alpha$-LMP algorithm. In this section we focus on covering problems whose IP matrix $A$ is totally balanced. More specifically, we study the integrality gap of the standard linear relaxation for Partial Totally Balanced Cover (P-TBC) shown below.

\begin{theorem} \label{theorem:LP-gap} Let $\IP$ and $\LP$ be the cost of the optimal integral and fractional solutions of an instance of P-TBC. Then $\IP \leq \left(1 + \frac{1}{3^{k-1}}\right) \LP + k\, c_{\max}$ for any $k \in Z_+$. Furthermore, for any large enough $k \in Z_+$ the exists an instance where $\IP > \left(1 + \frac{1}{3^{k-1}} \right) \LP + \frac{k}{2}\, c_{\max}$.
\end{theorem}

\begin{center}
\begin{minipage}[t]{5cm}
\begin{gather*}
\begin{array}{r@{\hspace{.7ex}}l}
\multicolumn{2}{c}{\min \ \vXv{c}{x} } \\[0.4cm]
\hspace{2ex} \mXv{A}{x} + \mXv{I}{r} \geq & \vf{1} \\[0.2cm]
\vXv{p}{r} \leq & p(U) - P  \\[0.2cm]
r_i, x_e \geq 0
\end{array}
\end{gather*}
\end{minipage}
\begin{pspicture}(0,1.75)(2.5,1.75)
\psline[linewidth=1.5pt,arrowinset=0,arrowlength=1]{->}(2.5,0)
\rput(1.25,-.4){\footnotesize LP Duality}
\end{pspicture}
\begin{minipage}[t]{6cm}
\begin{gather*}
\begin{array}{r@{\hspace{.7ex}}l}
  \multicolumn{2}{c}{\max\ \vXv{1}{y}  - (p(U) - P)\, \lambda  }\\[0.4cm]
\hspace{5ex}\mXv{A^T }{y}  & \leq \vf{c} \hspace{0.5cm} \\[0.2cm]
\vf{y} & \leq  \lambda \vf{p} \\[0.2cm]
y_i, \lambda & \geq 0 \\[0.4cm]
\end{array}
\end{gather*}
\end{minipage}
\end{center}

\ifthenelse{\boolean{journal}}{The rest of this section is devoted to proving Theorem~\ref{theorem:LP-gap}}{The rest of this section is devoted to proving the upper bound in Theorem~\ref{theorem:LP-gap}, the lower bound is left for the full version of the paper.} Our approach is based on Lagrangian relaxation and Kolen's algorithm for Prize-Collecting Totally Balanced Cover (PC-TBC). The latter exploits the fact that a totally balanced matrix can be put into greedy standard form by permuting the order of its rows and columns; in fact, the converse is also true \cite{HKS85}. A matrix is in standard greedy form if it does not contain as an induced submatrix 
\begin{equation}
\left[\begin{array}{cc} 1 & 1 \\ 1 & 0 \end{array}\right] \label{eq:forbidden-matrix}
\end{equation}
There are polynomial time algorithms that can transform a totally balanced matrix into greedy standard form 
\cite{S93} by shuffling the rows and columns of $A$. Since this transformation does not affect the underlying covering problem, we assume that $A$ is given in standard greedy form.

\subsection{Kolen's algorithm for Prize-Collecting Totally Balanced Cover}

For the sake of completeness we describe Kolen's primal-dual algorithm for PC-TBC. The algorithm finds a dual solution $\vf{y}$ and a primal solution $C$, which is then pruned in a reverse-delete step to obtain the final solution $\pruned{C}$. The linear and dual relaxations for PC-TBC appear below.
\begin{center}
\vspace{-0.5em}
\begin{minipage}[t]{5cm}
\begin{gather*}
\begin{array}{r@{\hspace{.7ex}}l}
\multicolumn{2}{c}{\min \ \vXv{c}{x} +  \lambda \vXv{p}{r}} \\[0.4cm]
\hspace{2ex} \mXv{A}{x} + \mXv{I}{r} \geq & \vf{1} \\[0.2cm]
r_i, x_e \geq 0
\end{array}
\end{gather*}
\end{minipage}
\begin{pspicture}(0,1.75)(2.5,1.75)
\psline[linewidth=1.5pt,arrowinset=0,arrowlength=1]{->}(2.5,0)
\rput(1.25,-.4){\footnotesize LP Duality}
\end{pspicture}
\begin{minipage}[t]{5cm}
\begin{gather*}
\begin{array}{r@{\hspace{.7ex}}l}
  \multicolumn{2}{c}{\max\ \vXv{1}{y}} \\[0.4cm]
\hspace{3ex}\mXv{A^T }{y}  & \leq \vf{c} \hspace{0.5cm} \\[0.2cm]
\vf{y} & \leq  \lambda \vf{p} \\[0.2cm]
y_i & \geq 0 \\[0.4cm]
\end{array}
\end{gather*}
\end{minipage}
\end{center}

The residual cost of the set $j$ w.r.t. $\vf{y}$ is defined as $c'_j = c_j - \sum_{i : a_{ij} = 1 } y_i$. The algorithm starts from the trivial dual solution $\vf{y = 0}$, and processes the elements in increasing column order of $A^T$. Let $i$ the index of the current element. Its corresponding dual variable, $y_i$, is increased until either the residual cost of some set $j$ containing $i$ equals 0 (we say set $j$ becomes tight), or $y_i$ equals $\lambda p_i$ (Lines 3-5).

\begin{figure}
\begin{center}
\small
\fancybox{
\begin{minipage}[t]{5cm}
\begin{algorithm}{Kolen}{(\mf{A}, \vf{c}, \vf{p}, \lambda)}
\mbox{// Dual update} \\
\vf{y} \=  \vf{0},\ C \= \emptyset,\ \pruned{C} \= \emptyset \\
 \begin{FOR}{i \= 1 \TO n }
 \delta \= \min \{ c'_j\, |\, a_{ij} = 1 \}\\
 y_i \leftarrow \min \{ \lambda p_i, \delta \}
 \end{FOR} \\
 C \= \{ j \, | \, c'_j = 0 \}
\end{algorithm}  \vspace{-1em}
\end{minipage}
\hspace{5ex}
\begin{minipage}[t]{5cm}
\begin{algorithm}{\global\algline=6}{}
\mbox{// Reverse delete} \\
 \begin{WHILE}{C \neq  \emptyset}
 j \= \mbox{largest set index in $C$} \\
 \pruned{C} \= \pruned{C} + j \\
 C \= C \setminus \{\, j'\, |\, \mbox{$j$ dominates $j'$ or $j= j'$}\, \}
 \end{WHILE} \\
\RETURN{(\pruned{C}, \vf{y})}
\end{algorithm}  \vspace{-1em}
\end{minipage}
} 
\end{center}
\end{figure}

Let $C = \{ j \, | \, c'_j =0 \}$ be the set of tight sets after the dual update is completed. As it stands the cover $C$ may be too expensive to be accounted for using the lower bound provided by $\vXv{1}{y}$ because a single element may belong to multiple sets in $C$. The key insight is that some of the sets in $C$ are redundant and can be pruned. 

\begin{definition}
 Given sets $j_1, j_2$ we say that $j_1$ \emph{dominates} $j_2$ in $\vf{y}$ if $j_1 > j_2$ and there exists an item $i$ such that $y_i > 0$ and $i$ belongs to $j_1$ and $j_2$, that is, $a_{ij_1} = a_{ij_2} = 1$. 
\end{definition}

The reverse-delete step iteratively identifies the largest index $j$ in $C$, adds $j$ to $\pruned{C}$, and removes $j$ and all the sets it dominates. This is repeated until no set is left in $C$ (Lines 8--11).

Notice that all sets $j \in C$ are tight, thus we can pay for set $j$ by charging the dual variables of items that belong to $j$. Because of the reverse-delete step if $y_i > 0$ then $i$ belongs to at most one set in $\pruned{C}$; thus in paying for $\pruned{C}$ we charge covered items at most once. Using the fact $A$ is in standard greedy form, it can be shown \cite{K82} that if $i$ was left uncovered then we can afford its penalty, i.e., $y_i = \lambda p_i$. The solution $\pruned{C}$ is optimal for PC-TBC since
\begin{equation} \label{eq:1-LMP}
\sum_{j \in \pruned{C}} \ c_j \  + \hspace{-2ex}  \sum_{ \substack{i \in U\ \mathrm{s.t.} \\ \nexists\, j \in \pruned{C} \, :\,  a_{ij} = 1} } \hspace{-2ex} \lambda p_i \ =\  \hspace{-2ex}  \sum_{ \substack{i \in U \ \mathrm{s.t.} \\ \exists\, j \in \pruned{C} \, :\,  a_{ij} = 1} } \hspace{-2ex} y_i +  \hspace{-2ex}  \sum_{ \substack{i\in U\ \mathrm{s.t.} \\ \nexists\, j \in \pruned{C} \, :\,  a_{ij} = 1} } \hspace{-2ex} y_i \ = \ \sum_{i \in U} \ y_i .
\end{equation}

If we could find a value of $\lambda$ such that {\scshape Kolen}$(A,c,p,\lambda)$ returns a solution $(\pruned{C}, \vf{y})$ covering \emph{exactly}~$P$ profit, we are done since from \eqref{eq:1-LMP} it follows that
\begin{equation}
\sum_{j \in \pruned{C}} c_j = \sum_{i \in U} y_i - \lambda\, (p(U) - P). \label{eq:DL-P-TBC}
\end{equation}
Notice that $(\vf{y},\lambda)$ is a feasible for the dual relaxation of P-TBC and its cost is precisely the right hand side of \eqref{eq:DL-P-TBC}. Therefore for this instance IP=DL=LP and Theorem~\ref{theorem:LP-gap} follows.

Unfortunately, there are cases where no such value of $\lambda$ exists. Nonetheless, we can always find a \emph{threshold value} $\lambda$ such that for any infinitesimally small $\delta>0$, $\lambda^-= \lambda - \delta$ and $\lambda^+=\lambda + \delta$ produce solutions covering less and more than $P$ profit respectively.
A threshold value can be found using Megiddo's parametric search \cite{M78} by making $O( n \log m )$ calls to the procedure {\scshape Kolen}. \ifthenelse{\boolean{journal}}{For completeness the technique is sketched in Appendix~\ref{section:thresholdvalue}.}{}

Let $\vf{y}$ ($\vf{y^-}$) be the dual solution and $C$ ($C^-$) the set of tight sets when {\scshape Kolen} is run on $\lambda$ ($\lambda^-$). Without loss of generality assume $\pruned{C}$ covers more than $P$ profit. (The case where $\pruned{C}$ covers less than $P$ profit is symmetrical: we work with $\vf{y}^+$ and $C^+$ instead of $\vf{y}^-$ and $C^-$.)

Our plan to prove Theorem~\ref{theorem:LP-gap} is to devise an algorithm to merge $\pruned{C}$ and $\pruned{C}^-$ in order to obtain a cheap solution covering at least $P$ profit. 

\subsection{Merging two solutions} \label{section:merging}

Before describing the algorithm we need to establish some important properties regarding these two solutions and their corresponding dual solutions.

For any $i$, the value of $y^-_i$ is a linear function of $\delta$ for all $i$. This follows from the fact that $\delta$ is infinitesimally small. Furthermore, the constant term in this linear function is $y_i$.

\begin{lemma} \label{lemma:tiny-difference} For each $i \in U$ there exists $a \in Z$, independent of $\delta$, such that $y^-_i = y_i + a \delta$.
\end{lemma}
\begin{proof}
By induction on the number of iteration of the dual update step of {\sc kolen}, using the fact that the same property holds for the residual cost of the sets.
\end{proof}

A useful corollary of Lemma~\ref{lemma:tiny-difference} is that $C^- \subseteq C$, since if the residual cost of a set is non-zero in $\vf{y}$ it must necessarily be non-zero in $\vf{y^-}$. The other way around may not hold.

At the heart of our approach is the notion of a merger graph $G=(V, E)$. The vertex set of $G$ is made up of sets from the two solutions, i.e., $V = \pruned{C} \oplus \pruned{C}^-$. The edges of $G$ are directed and given by
\begin{equation} \label{eq:edges-merger-graph}
E = \left\{
\hspace{1ex} (j_1, j_2) \hspace{1em}
\begin{array}{|@{\hspace{1em}}l@{\hspace{1ex}}}
j_1 \in \pruned{C}^- \setminus \pruned{C},\, j_2 \in \pruned{C} \setminus \pruned{C}^- \mbox{ s.t. $j_1$ dominates $j_2$ in $\vf{y}^-$, or } \\
j_1 \in \pruned{C} \setminus \pruned{C}^-,\, j_2 \in \pruned{C}^- \setminus \pruned{C} \mbox{ s.t. $j_1$ dominates $j_2$ in $\vf{y}$} \\
\end{array}
\right\}
\end{equation}

This graph has a lot of structure that can be exploited when merging the solutions.

\begin{lemma} The merger graph $G=(V,E)$ of $\pruned{C}^-$ and $\pruned{C}$ is a forest of out-branchings.
\end{lemma}

\begin{proof}
 First note that $G$ is acyclic, since if $(j_1, j_2) \in E$ then necessarily $j_1 > j_2$. Thus, it is enough to show that the in-degree of every $j \in V$ is at most one.
Suppose otherwise, that is, there exist $j_1, j_2 \in V$ such that $(j_1,j), (j_2,j) \in E$. Assume that $j_1 < j_2$ and $j \in \pruned{C}$ (the remaining cases are symmetrical).

\begin{floatingfigure}[r]{4.5cm}
\vspace{1.5ex}
\hspace{-3.5ex}
\begin{minipage}{4.25cm}
$\begin{array}{ccc@{\hspace{2ex}}|@{\hspace{2ex}}cc}
    & i_1 & i_2 & i_2 & i_1 \\
j   & 1   &  1 &  1  & 1 \\
j_1 & 1   &  \psframebox{1} &  & 1 \\
j_2 &     &  1 & 1 & \psframebox{1} 
\end{array}$
\end{minipage}
 \vspace{1ex}
\end{floatingfigure}

By definition \eqref{eq:edges-merger-graph}, we know that $j_1\, (j_2)\! \in \pruned{C}^-$ and that there exists $i_1$ ($i_2$) that belongs to $j$ and $j_1$ ($j_2$) such that $y^-_{i_1}> 0$ ($y^-_{i_2}> 0$). Since $A^T$ is in standard greedy form we can infer that $i_2$ belongs to $j_1$ if $i_1 < i_2$, or $i_1$ belongs to $j_2$ if $i_1> i_2$: The diagram on the right shows how, using the fact that $A^T$ does not contain~\eqref{eq:forbidden-matrix} as an induced submatrix, we can infer that the boxed entries must be 1. In either case we get that $j_2$ dominates $j_1$ in $\vf{y}^-$, which contradicts the fact that both belong to $\pruned{C}^-$.
\end{proof}

\begin{figure}[h]
\begin{center}
\small
\fancybox{
\begin{minipage}{7cm}
\begin{algorithm}{merge}{(\pruned{C}^-, \pruned{C})} 
\mbox{let $G$ be the merger graph for $\pruned{C}^-$ and $\pruned{C}$} \\
D \= \pruned{C}^- \\
 \begin{FOR}{\mbox{each root $r$ in $G$ } }
 \begin{IF}{p(D \oplus \subtree{r}) \leq P}
    \algkey{then}  D \= D \oplus \subtree{r} \\
    \algkey{else} \RETURN \CALL{increase}(r,D)
 \end{IF}
 \end{FOR}
\end{algorithm}
\vspace{-2ex}
\end{minipage}

}
\end{center}
\end{figure}

The procedure {\scshape merge} starts from the unfeasible solution $D = \pruned{C}^-$ and guided by the merger graph~$G$, it modifies $D$ step by step until feasibility is attained. The operation used to update $D$ is to take the symmetric difference of $D$ and a subtree of $G$ rooted at a vertex $r \in V$, which we denote by $\subtree{r}$. For each root $r$ of an out-branchings of $G$ we set $D \leftarrow D \oplus \subtree{r}$, until $p(D \oplus \subtree{r}) > P$. At this point we return the solution produced by {\scshape increase}$(r,D)$.

Notice that after setting $D \leftarrow D \oplus \subtree{r}$ in Line 5, the solution $D$ ``looks like'' $\pruned{C}$ within $\subtree{r}$. Indeed, if all roots are processed then $D = \pruned{C}$. Therefore, at some point we are bound to have $p(D \oplus \subtree{r}) > P$ and to make the call {\scshape increase}$(r,D)$ in Line 6.
Before describing {\scshape increase} we need to define a few terms. Let the \emph{absolute benefit} of set $j$, which we denote by $b_j$, be the profit of elements uniquely covered by set $j$, that is,
\begin{equation}
b_j = p \left( \big\{\ i\in U \  \suchthat\  \forall\, j' \in \pruned{C} \cup \pruned{C}^- : a_{ij'} = 1 \mbox{ iff } j'=j\ \big\}\right).
\end{equation}
Let  $D \subseteq \pruned{C} \cup \pruned{C}^-$. Note that if $j\in D$, the removal of $j$  decreases the profit covered by $D$ by at least $b_j$; on the other hand, if $j \notin D$, its addition increases the profit covered by at least $b_j$. This notion of benefit can be extended to subtrees,
\begin{equation}
\benefit{T_j}{D} = \sum_{j' \in \subtree{j} \setminus D} b_{j'} - \sum_{j' \in \subtree{j} \cap D} b_{j'}.
\end{equation}
We call this quantity the \emph{relative benefit} of $\subtree{j}$ with respect to $D$. It shows how the profit of uniquely covered elements changes when we take $D \oplus \subtree{j}$. Note that $\benefit{\subtree{j}}{D}$ can positive or negative.

Everything is in place to explain {\scshape increase}$(j,D)$. The algorithm assumes the input solution is unfeasible but can be made feasible by adding some sets in $\subtree{j}$; more precisely, we assume $p(D) \leq P$ and $P < p(D) + \benefit{\subtree{j}}{D}$. If adding $j$ to $D$ makes the solution feasible then return $D+j$ \mbox{(Lines 2-3)}. If there exists a child $c$ of $j$ that can be used to propagate the call down the tree then do that \mbox{(Lines 4-5)}. Otherwise, \emph{split} the subtree $\subtree{j}$: Add $j$ to $D$ and process the children of $c$, setting $D \leftarrow D \oplus \subtree{c}$ until $D$ becomes feasible (Lines 6-9). At this point $p(D) > P$ and $p(D \oplus \subtree{c}) \leq P$. If $P - p(D \oplus \subtree{c})< p(D) - P $ then call {\scshape increase}$(c,D \oplus \subtree{c})$ else call {\scshape decrease}$(c,D)$ and let $D'$ be the cover returned by the recursive call (Lines 10-12). Finally, return the cover with minimum cost between $D$ and $D'$.

\begin{figure}[h]
\small
\fancybox{
\begin{minipage}[t]{5cm}
\begin{algorithm}{increase}{(j, D)}
\mbox{// assume $ p(D) \leq P < p(D) + \benefit{\subtree{j}}{D}$ } \\
\begin{IFTHEN}{p(D + j) \geq P} \RETURN D + j \end{IFTHEN} \\
\begin{IFTHEN}{\exists \mbox{ child $c$ of $j$}\,:\, p(D) + \benefit{\subtree{c}}{D} > P } \RETURN{\CALL{increase}(c,D)} \end{IFTHEN} \\
D \= D + j \\
\begin{WHILE}{\mbox{$p(D) \leq P$}}
c \= \mbox{child of $j$ maximizing } \benefit{\subtree{c}}{D} \\
D \= D \oplus \subtree{c}
\end{WHILE} \\
\begin{IF}{ P - p(D \oplus \subtree{c})< p(D) - P  }
\algkey{then} D' \= \CALL{increase}(c,D \oplus \subtree{c}) \\
\algkey{else} D' \= \CALL{decrease}(c,D)
\end{IF} \\
\RETURN \mbox{min cost } \{D, D'\}
\end{algorithm} \vspace{-1em}
\end{minipage}
\begin{minipage}[t]{5cm}
\begin{algorithm}{decrease}{(j, D)}
\mbox{// assume $ p(D) \geq P> p(D) + \benefit{\subtree{j}}{D}$ } \\
\begin{IFTHEN}{p( (D \oplus \subtree{j}) + j) \geq P} \RETURN (D \oplus \subtree{j}) + j \end{IFTHEN} \\
\begin{IFTHEN}{\exists \mbox{ child $c$ of $j$}\,:\, p(D) + \benefit{\subtree{c}}{D} < P } \RETURN{\CALL{decrease}(c,D)} \end{IFTHEN} \\
D \= D + j \\
\begin{WHILE}{\mbox{$p(D) \geq P$}}
c \= \mbox{child of $j$ minimizing } \benefit{\subtree{c}}{D} \\
D \= D \oplus \subtree{c}
\end{WHILE} \\
\begin{IF}{  p(D \oplus \subtree{c}) - P < P - p(D)}
 \algkey{then} D' \= \CALL{increase}(c,D) \\
 \algkey{else} D' \= \CALL{decrease}(c,D \oplus \subtree{c})
\end{IF} \\
\RETURN \mbox{min cost } \{D \oplus \subtree{c}, D'\}
\end{algorithm} \vspace{-1em}
\end{minipage}
\hspace{-2ex}
} 
\end{figure}

The twin procedure {\scshape decrease}$(j,D)$ is essentially symmetrical: Initially the input is feasible but can be made unfeasible by removing some sets in $\subtree{j}$; more precisely $p(D)\geq P$ and  $P< p(D) + \benefit{\subtree{c}}{D}$.

At a very high level, the intuition behind the {\scshape increase/decrease} scheme is as follows. In each call one of three things must occur:
\begin{list}{}{\setlength{\leftmargin}{5em}\setlength{\labelsep}{2ex}\setlength{\topsep}{0.75ex}\setlength{\itemsep}{0.5ex}\setlength{\parsep}{0cm}}
\item[(i)] A feasible cover with a small coverage excess is found (Lines 2-3), or
\item[(ii)] The call is propagated down the tree at no cost (Lines 4-5), or
\item[(iii)] A subtree $T_j$ is split (Lines 6-9). In this case, the cost $c_j$ cannot be accounted for, but the offset in coverage $|P-p(D)|$ is reduced at least by a factor of 3.
\end{list}
If the {\scshape increase/decrease} algorithms split many subtrees (incurring a high extra cost) then the offset in coverage must have been very high at the beginning, which means the cost of the dual solution is high and so the splitting cost can be charged to it. In order to flesh out these ideas into a formal proof we need to establish some crucial properties of the merger graph and the algorithms. \ifthenelse{\not \boolean{journal}}{Proofs are omitted due to lack of space.} {}

\begin{lemma} \label{lemma:white-coverage} 
If $y_i < \lambda\, p_i$ then there exist $j' \in \pruned{C}$ and $j'' \in \pruned{C}^-\!$ such that either $j'\! =\! j''$ or $(j', j'') \in E$ or $(j'', j') \in E$.
\end{lemma}

\ifthenelse{\boolean{journal}}{
\mbox{}\vspace{-1em}

\begin{floatingfigure}[r]{5.5cm}
\vspace{-1ex}
\hspace{-2.5ex}
\begin{minipage}{5cm}
\psset{framesep=2pt}
$\begin{array}{cccc@{\hspace{2ex}}|@{\hspace{2ex}}ccc}
    & i_1 & i_2 & i & i_2 & i_1 & i \\
j   & 1   &  1 &  1  & 1 & 1 & 1 \\
j_1 & 1   &  \psframebox{1} & \psframebox{1} &  & 1 & \psframebox{1} \\
j_2 &     &  1 & \psframebox{1} & 1 & \psframebox{1}  & \psframebox{1} 
\end{array}$
\end{minipage}
\vspace{1ex}
\end{floatingfigure}

\noindent {\it Proof.} Since $y_i < \lambda\, p_i$, by Lemma~\ref{lemma:tiny-difference} we get that $y^-_i < \lambda\, p_i$ as well. Thus, there exists a set $j \in C^-$ such that $a_{ij}=1$ that becomes tight right after processing $i$ in the dual update of {\sc kolen}. Due to the reverse-delete step either $j \in \pruned{C}^-$ or there exists $j_1 \in \pruned{C}^-$ that dominates $j$ in $\vf{y^-}$ trough some element $i_1 \in U$, i.e., $y_{i_1} > 0$ and $a_{i_1 j} = a_{i_1 j_1} = 1$. In the latter case, since the set $j$ is tight after $i$ is processed, it follows that $i_1 \leq i$. Because $A$ is in standard greedy form we infer that $a_{i j_1} = 1$.
By Lemma~\ref{lemma:tiny-difference}, we have $C^- \subseteq C$, thus $j \in C$. A similar reasoning as above shows that either $j \in \pruned{C}$ or there exists $j_2 \in \pruned{C}$ that dominates $j$ in $\vf{y}$ trough some element $i_2 \in U$ such that $a_{i j_2} = 1$ and $i_2 \leq i$.

If there exists a set in $\pruned{C}^- \cap \pruned{C}$ covering $i$ the lemma follows, so suppose otherwise. If $j \in \pruned{C}^-$ ($j \in \pruned{C}$) then there exists $j_2\in \pruned{C}$ $(j_1 \in \pruned{C}^-$) that dominates $j$ in $\vf{y^-}$ ($\vf{y}$), and again the lemma holds.
Finally, consider the case $j \notin \pruned{C}^-$ and $j \notin \pruned{C}$. Assume $j_1 < j_2$, the other case is symmetrical.
Because $A$ is in standard greedy form we get that $a_{i_2j_1} = 1$ and $(j_2, j_1) \in E$ if $i_1 < i_2$, or $a_{i_1j_2}=1$ and $(j_1, j_2) \in E$ if $i_2 < i_1$. In either case the lemma follows.
The diagram on the right shows a summary of the entries that were inferred using the fact that $A$ does not contain~\eqref{eq:forbidden-matrix} as an induced submatrix. \qed
}{}

 \begin{lemma} \label{lemma:alternating} Let $(j,D)$ be the input of {\sc increase/decrease}. Then at the beginning of each call we have $j' \in D$ or $j'' \in D$ for all $(j',j'') \in E$. Furthermore, if $j' \in D$ and $j'' \in D$ then $j'$ or $j''$ must have been split in a previous call. 
 \end{lemma}

 \ifthenelse{\boolean{journal}}
 { 
 \begin{proof}
We prove the lemma by induction on the number of calls to {\sc increase/decrease}. Clearly, the lemma is true on the first call when $|\{j',j''\} \cap D| = 1$ for all $(j',j'') \in E$ and no vertex has been split yet.

 Suppose the lemma holds at the beginning of a call to {\sc increase/decrease}, we argue that it also holds at the beginning of the next call. If the next call happens in Line 5 then the lemma holds trivially since $D$ did change. Otherwise, $j$ is split; namely, $j$ is added to $D$ (Line 6) and $D$ is updated by taking $D \oplus \subtree{c}$ (Line 9) for a number of children $c$ of $j$. Before splitting $j$, for every child $c$ of $j$ and every $(j' j'') \in \subtree{c}$ we have $|\{j',j''\} \cap D| = 1$, and the same holds after $j$ is split. Hence, the lemma follows.
  \end{proof}
 }
 {}

\begin{lemma} \label{lemma:precondition} Let $(j,D)$ be the input of {\sc increase/decrease}. For {\sc increase} we always have $p(D) \leq P < p(D) + \benefit{\subtree{j}}{D}$, and for {\sc decrease} we have $p(D) \geq  P > p(D) + \benefit{\subtree{j}}{D}$.
\end{lemma}

\ifthenelse{\boolean{journal}}
{
\begin{proof}
By induction on the number of calls to {\sc increase/decrease}. At the base case the call is made by {\sc merge} and thus we have $p(D) \leq P < p(D \oplus \subtree{j})$. We claim  that
\begin{equation}
\label{eq:coverage-change} p(D \oplus \subtree{j}) - p(D) = \benefit{\subtree{j}}{D}, 
\end{equation}
from which the lemma follows. We argue that every $i \in U$ contributes the same amount to each side of \eqref{eq:coverage-change}. From now on we assume that $i$ is not covered by $D \setminus T_j$, otherwise its contribution to \eqref{eq:coverage-change} is zero.

Suppose $i$ is uniquely covered by some set $j \in T_j$. Then its contribution to both sides of \eqref{eq:coverage-change} is either $p_i$ or $-p_i$ depending on whether $i \in D \oplus \subtree{j}$ or $i \in D$.

Now consider the case when $i$ is covered by multiple sets in $T_j$. Recall that $|\{j',j''\} \cap D|=1$ for all $(j',j'') \in T_j$ because $D \cap T_j = \pruned{C}^- \cap T_j$. If $y_i < \lambda p_i$ then by Lemma~\ref{lemma:white-coverage} is covered by both $D$ and $D \oplus \subtree{j}$. If $y_i = \lambda p_i$ then $i$ covered by at most one set in $\pruned{C}^-$ and at most one set in $\pruned{C}$. It follows that $i$ must be covered exactly by one set in $D$ and another set in $D \oplus \subtree{j}$. Hence, the contribution of $i$ to both sides of \eqref{eq:coverage-change} is zero.

For the inductive step suppose the lemma holds at the beginning of this call. Clearly, if the next call is made in Line~3 the lemma holds. Suppose that the call is made in Lines 11-12. Note that after adding $j$ to $D$ (Line 7) we have for every child $c$ of $j$
\[p(D \oplus \subtree{c}) - p(D) = \benefit{\subtree{c}}{D}.\]
Therefore, by inductive hypothesis, we are bound to exit the while loop and the lemma holds in the next call.
\end{proof}
}{}

Recall that $\vf{y}$ is also a feasible solution for the dual relaxation of P-TBC and its cost is given by $\DL = \sum_{i=1}^n y_i - (p(U) - P) \lambda$. The following lemma proves the upper bound of Theorem~\ref{theorem:LP-gap}.

\begin{lemma} \label{lemma:merge} Suppose {\scshape merge} outputs $D$. Then $c(D) \leq \big(1+ \frac{1}{3^{k-1}}\big) \DL + k\, c_{\max}$ for all $k \in Z_+$.
\end{lemma}

\begin{proof}
Let us digress for a moment for the sake of exposition. Suppose that in Line 6 of {\sc merge}, instead of calling {\scshape increase}, we return $D'\! =\! D \oplus \subtree{r}$. Notice every arc in the merger graph has exactly one endpoint in $D'$. By Lemma~\ref{lemma:white-coverage}, any element $i$ not covered by $D'$ must have $y_i = \lambda\, p_i$.
Furthermore, if $y_i >0$ then there exists at most one set in $D'$ that covers $i$; if two such sets exist, one must dominate the other in $\vf{y}$ and $\vf{y}^-$, which is not possible. Hence,
\begin{equation} \label{eq:DL-bound}
c(D) = \sum_{j \in D'} \sum_{i: a_{ij}=1} y_i = \hspace{-2ex} \sum_{ \substack{i\ \mathrm{s.t.} \\ \exists\, j \in D' \, :\,  a_{ij} = 1}} \hspace{-2ex} y_i\, = \sum_{i \in U} \ y_i - (p(U)-p(D')) \lambda \leq \DL + (p(D') - P) \lambda 
\end{equation}
In the fortunate case that $(p(D') - P) \lambda \leq k c_{\max}$, the lemma would follow. Of course, this need not happen and this is why we make the call to {\scshape increase} instead of returning $D'$.

Let $j_q$ be the root of the $q^{\mathrm{th}}$ subtree split by {\sc increase/decrease}. Also let $D_q$ the solution right before splitting $\subtree{j_q}$, and $D'_q$ and $D''_q$ be the unfeasible/feasible pair of solutions after the splitting, which are used as parameters in the recursive calls (Lines 11-12). Suppose Lines 7-9 processed only one child of $j_q$, this can only happen in {\sc increase}, in which case $p(D''_q)> P$ but $p(D''_q) - b_{j_q} < P$. The same argument used to derive \eqref{eq:DL-bound} gives us
\begin{equation}
c\left(D''_q \setminus \{j_{\leq q}\}\right) \leq \sum_{i \in U} \ y_i - \left(p(U)-p(D''_q) + b_{j_q}\right) \lambda  \leq \DL  
\end{equation}
The cost of the missing sets is $c(\{j_{\leq q}\}) \leq q\, c_{\max}$, thus if $q\leq k$ the lemma follows. A similar bound can be derived if the recursive call ends in Line 3 before splitting the $k^{\mathrm{th}}$ subtree.
Finally, the last case to consider is when Lines 7-9 process two or more children $j_q$ for all $q \leq k$. In this case
\begin{equation} |p(D_q) - P| \geq 3 \min \left\{ |p(D'_q) -P|, |p(D''_q) -P|\right\} = 3\, |p(D_{q+1}) - P|,
\end{equation}
which implies $|p(D_1) -P| \geq 3^{k-1} |p(D_k) -P|\geq 3^{k-1} | p(D''_k) - P|$. Also, $\lambda ( P - p(D_1)) \leq \DL$ since all elements $i$ not covered by $D_1$ must be such that $y_i = \lambda p_i$. Hence, as before
\begin{equation}
c\left(D''_k \setminus \{j_{\leq k}\}\right) \leq \DL + \lambda \left(p(D''_k) - P\right)  \leq \DL + \lambda \frac{P-p(D_1)}{3^{k-1}} \leq \left(1 + \frac{1}{3^{k-1}}\right) \DL
\end{equation}
Adding the cost of $\{j_{\leq k}\}$ we get the lemma.
\end{proof}


\ifthenelse{\boolean{journal}}{
\subsection{Integrality gap example} \label{section:gap-example}

To finish the proof of Theorem~\ref{theorem:LP-gap} we now show a family of instances of P-TBC exhibiting an integrality gap of $\IP > \left(1 + \frac{1}{3^{k-1}}\right) \LP + \frac{k}{2} c_{\max}$ for large enough $k$

Let $T$ be a rooted tree and $\{ (s_i, t_i) \}_{i=1}^n$ be a collection of pairs of nodes of $T$, each defining a unique path in $T$. Let $A$ be the incidence matrix of paths to edges of $T$. The covering problem defined by $A$ is the well-know Multicut problem where the objective is to find a minimum cost set of edges whose removal separates all pairs. If for every $(s_i,t_i)$ pair $s_i$ is the ancestor of $t_i$ or vice-versa then $A$ is totally balanced.

\newcommand{\edge}[1]{\psline{o-o}#1}
\newcommand{\chosenedge}[1]{\edge{#1}\pszigzag[linewidth=0.5pt,coilwidth=2.5pt,coilarm=2pt]{o-o}#1}
\psset{unit=12pt}

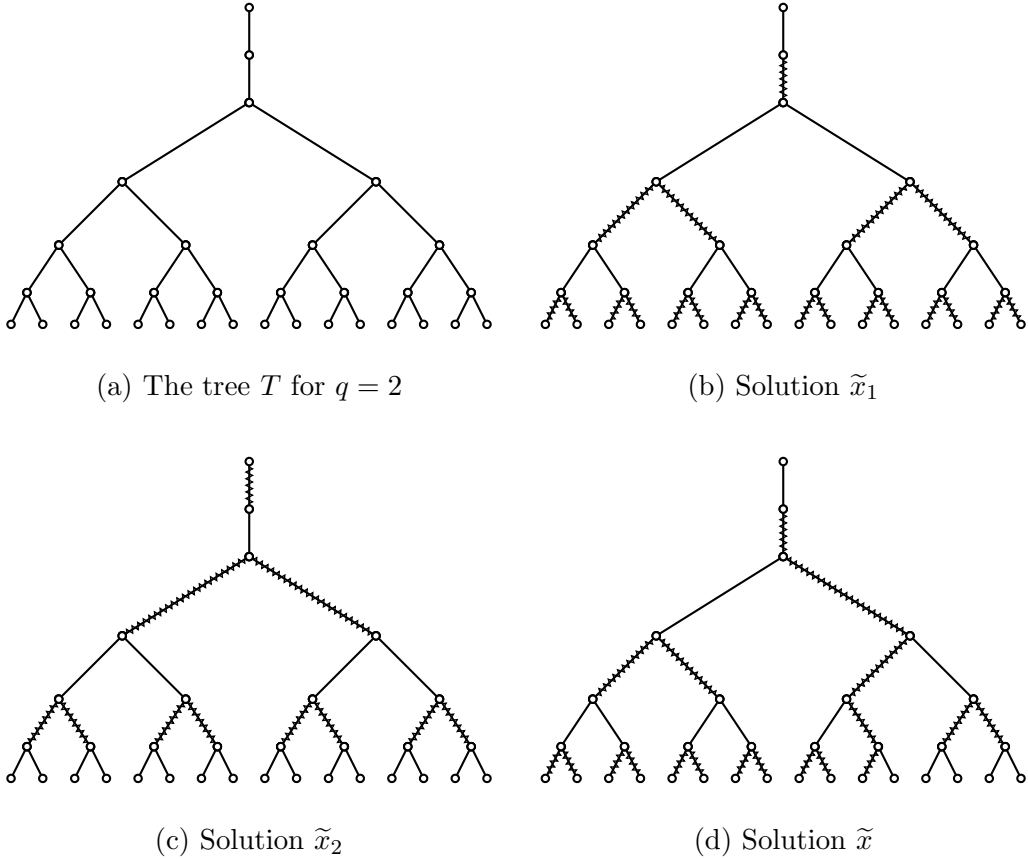
\begin{figure}
\[ \begin{array}{cc}
\begin{pspicture}(0,-2)(16,12)
\rput(7.5,-1){(a) The tree $T$ for $q=2$}
\multirput(0,1)(2,0){8}{\edge{(0,0)(0.5,1)}\edge{(1,0)(0.5,1)}}
\multirput(0.5,2)(4,0){4}{\edge{(0,0)(1,1.5)}\edge{(2,0)(1,1.5)}}
\multirput(1.5,3.5)(8,0){2}{\edge{(0,0)(2,2)}\edge{(4,0)(2,2)}}
\rput(3.5,5.5){\edge{(0,0)(4,2.5)}\edge{(8,0)(4,2.5)}}
\rput(7.5,8){\edge{(0,0)(0,1.5)}\edge{(0,1.5)(0,3)}}
\end{pspicture}
&
\begin{pspicture}(0,-2)(16,12)
\rput(7.5,-1){(b) Solution $\widetilde{x}_1$}
\multirput(0,1)(2,0){8}{\chosenedge{(0,0)(0.5,1)}\chosenedge{(1,0)(0.5,1)}}
\multirput(0.5,2)(4,0){4}{\edge{(0,0)(1,1.5)}\edge{(2,0)(1,1.5)}}
\multirput(1.5,3.5)(8,0){2}{\chosenedge{(0,0)(2,2)}\chosenedge{(4,0)(2,2)}}
\rput(3.5,5.5){\edge{(0,0)(4,2.5)}\edge{(8,0)(4,2.5)}}
\rput(7.5,8){\chosenedge{(0,0)(0,1.5)}\edge{(0,1.5)(0,3)}}
\end{pspicture} \\

\begin{pspicture}(0,-2)(16,12)
\rput(7.5,-1){(c) Solution $\widetilde{x}_2$}
\multirput(0,1)(2,0){8}{\edge{(0,0)(0.5,1)}\edge{(1,0)(0.5,1)}}
\multirput(0.5,2)(4,0){4}{\chosenedge{(0,0)(1,1.5)}\chosenedge{(2,0)(1,1.5)}}
\multirput(1.5,3.5)(8,0){2}{\edge{(0,0)(2,2)}\edge{(4,0)(2,2)}}
\rput(3.5,5.5){\chosenedge{(0,0)(4,2.5)}\chosenedge{(8,0)(4,2.5)}}
\rput(7.5,8){\edge{(0,0)(0,1.5)}\chosenedge{(0,1.5)(0,3)}}
\end{pspicture}
&
\begin{pspicture}(0,-2)(16,12)
\rput(7.5,-1){(d) Solution $\widetilde{x}$}
\multirput(0,1)(2,0){5}{\chosenedge{(0,0)(0.5,1)}\chosenedge{(1,0)(0.5,1)}}
\rput(10,1){\chosenedge{(0,0)(0.5,1)}\edge{(1,0)(0.5,1)}}
\multirput(12,1)(2,0){2}{\edge{(0,0)(0.5,1)}\edge{(1,0)(0.5,1)}}
\multirput(0.5,2)(4,0){2}{\edge{(0,0)(1,1.5)}\edge{(2,0)(1,1.5)}}
\rput(8.5,2){\edge{(0,0)(1,1.5)}\chosenedge{(2,0)(1,1.5)}}
\rput(12.5,2){\chosenedge{(0,0)(1,1.5)}\chosenedge{(2,0)(1,1.5)}}
\rput(1.5,3.5){\chosenedge{(0,0)(2,2)}\chosenedge{(4,0)(2,2)}}
\rput(9.5,3.5){\chosenedge{(0,0)(2,2)}\edge{(4,0)(2,2)}}
\rput(3.5,5.5){\edge{(0,0)(4,2.5)}\chosenedge{(8,0)(4,2.5)}}
\rput(7.5,8){\chosenedge{(0,0)(0,1.5)}\edge{(0,1.5)(0,3)}}
\end{pspicture}
\end{array}\]
\caption{Integrality gap example for $q=2$. In (b-d) the wiggly edges belong to the corresponding solution \label{fig:example-gap}}

\end{figure}

Our tree $T$ is made up of a complete binary tree with height $2 q$ plus a 2-path coming out of the apex of the binary tree going up into the real root of $T$; thus the tree has $2^{2q+1}+1$ nodes. The reader is referred Figure~\ref{fig:example-gap} for a picture of the instance. The cost of every edge is 3. There are two kinds of paths: \emph{internal} and \emph{fringe} paths. For every node in the binary tree there is an internal path of length two coming out of the node going up; there are $2^{2q+1}-1$ such paths each having a profit of $4^q$. For each leaf and the root there is a fringe path of length 1 incident on it; there are $2^{2q}+1$ such paths each having a profit of 2. The target coverage is given by $\overline{P} = 2 \left( 4^{q-1}+\ldots + 4^0 + 1\right) = \frac{2^{2q+1}+ 4}{3}$, where the shorthand notation $\overline{X}$ stands for $p(U) - X$. 

Consider the dual solution $\vf{y}$ where every internal path gets a dual value of 1 and every fringe path gets a dual value of $2$. The solution is feasible for $\lambda = 1$ and has cost
\begin{equation}
\DL = \sum_{i \in U} y_i - \overline{P} \lambda = \frac{10 \, 4^q -1 }{3}
\end{equation}

To show that $\vf{y}$ is optimal, we construct a primal (fractional) solution $\vf{x}$ with the same cost, which is a convex combination of two integral solutions $\vf{\widetilde{x}_1}$ and $\vf{\widetilde{x}_2}$. Let $\vf{\widetilde{x}_1}$ consist of edges in every other level of~$T$ starting at the leaf level and let $\vf{\widetilde{x}_2}$ be the complement of $\vf{\widetilde{x}_1}$. (See Figure~\ref{fig:example-gap}. Note that $\overline{p(\vf{\widetilde{x}_1})} = 2$ and $\overline{p(\vf{\widetilde{x}_2})}= 2^{2q+1}$. Consider the convex combination $\alpha \overline{p(\vf{\widetilde{x}_1})} + \beta \overline{p(\vf{\widetilde{x}_2})} = \overline{P}$ and let $\vf{x} = \alpha \vf{\widetilde{x}_1} + \beta \vf{\widetilde{x}_2}$. Its cost is given by
\begin{eqnarray*}
c(\vf{x}) & = & \alpha c(\vf{\widetilde{x}_1}) + \beta c(\vf{\widetilde{x}_2}) \\
 &= & \alpha \left(\sum_i y_i - \lambda \overline{p(\vf{\widetilde{x}_1})} \right) + \beta \left(\sum_i y_i - \lambda \overline{p(\vf{\widetilde{x}_2})}\right) - \lambda \overline{P} \\
 &= &\DL.
\end{eqnarray*}
Clearly, $\vf{x}$ is a feasible fractional solution. Therefore, it is optimal. 

Let $\vf{\widetilde{x}}$ be the solution defined as follows. For edges incident on a leaf and leave out the $\frac{\overline{P}}{2}-1$ rightmost ones and choose remaining ones. For other edges, choose the edge only if one of the edges immediately below is not chosen. (See Figure~\ref{fig:example-gap}). If we try to pay for $\vf{\widetilde{x}}$ using the dual cost we will charge twice $2q-1$ internal paths whose both edges are chosen in $\vf{\widetilde{x}}$. In other words,
\begin{equation}
c(\vf{\widetilde{x}}) = \DL + 2q-1
\end{equation}

Due to their high profit, internal paths cannot be left uncovered. Using this fact we can infer that $\vf{\widetilde{x}}$ is indeed an optimal integral solution covering $P$ profit. Choosing $k = q\, \log_3 4$ we get the lower bound of Theorem~\ref{theorem:LP-gap}. That is, for large enough $q$,
\begin{equation} c(\vf{\widetilde{x}}) > \left(1+ \frac{1}{3^{k-1}}\right) c(x) + c_{\max} \frac{k}{2}.
\end{equation}

For smaller values of $k$ the slightly weaker bound with $c_{\max}\, \frac{k-5}{2}$ additive error holds. It is worth noting that the example can be adapted to yield the same bound for instances with unit profits.

\section{Applications} \label{sec:applications}

In this section we show how Theorem~\ref{theorem:LP-gap} implies better approximation algorithms for a number of covering problems that can be expressed with a suitable combination of $\rho$ totally-balanced matrices.

\begin{definition} Matrix $B$ is said to be row-induced by a collection of matrices $A_1, \ldots A_k \in R^{n \times m}$ if for all $i$, the $i$th row of $B$ equals the $i$th row of $A_j$ for some $1 \leq j \leq k$.
\end{definition}

\begin{definition} Matrix $A \in \{0,1\}^{n \times m}$ is said to be $\rho$-separable if there exist matrices $A_1, \ldots A_\rho \in \{0,1\}^{x \times m}$ such that  $A = \sum_q A_q$ and every matrix row-induced by $A_1, \ldots, A_\rho$ is totally balanced.
\end{definition}

Our algorithms make us of following lemma to absorb the additive error in our bounds.

\begin{lemma} Let $\myalgo{A}$ be an algorithm for a given partial covering problem $(U, \collection{S}, P)$ that produces a solution with cost at most $\alpha\, \OPT + k \, c_{\max}$, where $\OPT$ is the cost of the optimal solution. Then there exists an $\alpha$-approximation that makes $|U|^{\frac{k}{\alpha-1}}$ calls to $\myalgo{A}$. \label{lemma:absorb}\end{lemma}
\begin{proof}
The idea is to run $\myalgo{A}$ on a modified instance $(U', \collection{S}', P')$. Let $X$ be the $\frac{k}{\alpha-1}$ most expensive sets in the optimal cover for $(U,\collection{S}, P)$. Let $\collection{S}' = \collection{S} \setminus \{\, j\, | \, c_j > \min_{j' \in X} c_{j'} \}$, $U' = U \setminus \{\, i \, | \, \mbox{covered by } X\}$, and $P' = P - p(X)$. The optimal solution in the new instance has cost $\OPT' = \OPT - c(X)$.

Adding $X$ to the solution returned by $\myalgo{A}(U',S',P')$ gives us a feasible solution, for the original instance, with cost at most
\[\alpha\, \OPT' + \, k \, c'_{\max} + c(X) \leq \alpha\, \OPT' + k \, \frac{\alpha - 1}{k} c(X) + c(X) = \alpha\, \OPT.\]

Unfortunately we do not know which sets comprise $X$. Therefore $\myalgo{A}$ is run on every choice of $X$ and the best cover found is returned. The number of calls to $\myalgo{A}$ needed is $|U|$ choose ${\frac{k}{\alpha-1}}$.
\end{proof}

We are ready to describe our approximation algorithms for covering problems that can be described with a \mbox{$\rho$-separable} matrix. We assume the decomposition is given to us. For an arbitrary matrix finding such a decomposition, or even testing for its existence, may be hard. However, for our application problems it is easy to find the $\rho$ matrices using the problem definition.

\begin{theorem} \label{theorem:approx} Let $A$ be $\rho$-separable into matrices $A_1, \ldots, A_\rho$ where $\rho > 1$. For any constant $\epsilon > 0$ there is a $(\rho + \epsilon)$-approximation and a quasi-polynomial time $\rho$-approximation for the partial covering problem defined by $A$.
\end{theorem}
\begin{proof}
Our algorithm is based on the approach of \cite{GNS06,HS05a}. First, we find an optimal fractional solution $(\vf{x},\vf{r})$ for the partial covering problem defined by $A$. Let $\vf{a_i^q}$ be the $ith$ row of $A_q$. Notice that for each $i$ there must exist a $q_i$ such that $\vXv{a_i^{q_i}}{x} \geq \frac{1- r_i}{\rho}$. Second, we construct a matrix $B$ by choosing $\vf{a^{q_i}_i}$ as the $i$th row of $B$. Note that $B$ is totally balanced. Finally, we find a threshold value $\lambda^*$ for $B$ and invoke {\sc merge} to find a cover $D$.

Any feasible solution for $B$ is also feasible for $A$, thus $D$ is a feasible cover for $A$. Note that $(\rho \vf{x},\vf{r})$ is a feasible fractional solution for $B$. Let $\OPT$ be the cost optimal solution for $A$. By Lemma~\ref{lemma:merge} and letting $k \geq \log_3 \frac{\rho}{\epsilon} + 1$ we get,
\begin{equation}
c(D) \leq \left( 1 + \frac{1}{3^{k-1}} \right) c(\rho x) + k\, c_{\max} = \left(1 + \frac{1}{3^{k-1}}\right) \rho\, c(x) + k\, c_{\max} \leq  (\rho + \epsilon) \OPT + k\, c_{\max}
\end{equation}
This solution is $\rho + \epsilon$ approximate with an additive error of $k\, c_{\max}$ that can be absorbed using Lemma~\ref{lemma:absorb}

For the quasi-polynomial time $\rho$-approximation, setting $k \geq \log \rho\, |U| + 1$ we get 
\begin{equation}
c(D) \leq \left(1 + \frac{1}{\rho\, |U|}\right) c(\rho x) + k\, c_{\max} \leq \rho \,c(x) + (k+1)\, c_{\max} \leq \rho\, \OPT + (k+1) \ c_{\max}.
\end{equation}
And the theorem follows.
\end{proof}

This implies improved approximation algorithm for the partial version of Multicut and Path Hitting on Trees, and Rectangle Stabbing. To show this, we use the following fact about totally balanced matrices. Let $T$ be a rooted tree. An $s$-$t$ path in $T$ is said to be \emph{descending} if $s$ is an ancestor of $t$. Let $\mathcal{P}$ and $\mathcal{Q}$ be collections of descending paths in $T$, and let $A$ be the $\mathcal{P}$-$\mathcal{Q}$ incidence matrix $A=\{a_{i,j}\}$, where $a_{i,j} = 1$ if and only if the $i$th path in $\mathcal{P}$ intersects the $j$th path in $\mathcal{Q}$. It is known that $A$ is totally balanced: To put the matrix into Greedy Standard form arrange the columns and rows of $A$ so that the paths in $\mathcal{P}$ and $\mathcal{Q}$ appear in non-increasing distance from the root.

\begin{corollary} There is a $2 + \epsilon$ approximation and a quasi-polynomial time $2$-approximation for Partial Multicut on Trees.
\end{corollary}

\begin{proof} 
The input of Partial Multicut is a tree $T$ and a collection of paths $\mathcal{P}$ in $T$, the problem is defined by the $\mathcal{P}$-$E[T]$ incidence matrix $A$. Even though a path in $\mathcal{P}$ may not be descending, we can always split such a path into two separate descending paths. Therefore, $A$ is 2-separable.
\end{proof}

\begin{corollary} There is a $4 + \epsilon$ approximation and a quasi-polynomial time $4$-approximation for Path Hitting on Trees.
\end{corollary}
\begin{proof}
The input of Partial Path Hitting is a tree $T$ and two collections of paths $\mathcal{P}$ and $\mathcal{Q}$ in $T$, the covering problem is defined by the $\mathcal{P}$-$\mathcal{Q}$ incidence matrix $A$. \citet{PS06} noted that if we split each path in $\mathcal{Q}$ into two descending paths the cost of the optimal solution increases by at most a factor of 2. The matrix of this modified problem is 2-separable.
\end{proof}

Our last application problem is Rectangle Stabbing. This is a special case of Set Cover with $\rho$-Blocks, a broad class of covering problems introduced by \citet{HS05a}, where the incidence matrix defining the problem is such that every row has $\rho$ blocks of contiguous~\mbox{1's}.

\begin{theorem} Let $A$ be a matrix defining an instance of Set Cover with $\rho$-Blocks. Then there exists a polynomial time $\rho$-approximation algorithm for the partial covering problem defined by $A$.
\end{theorem}

\begin{proof} We proceed as in Theorem~\ref{theorem:approx} to reduce $A$ to a matrix $B$. This new matrix is not only totally balanced, but each row consists of a single block of consecutive 1's. For such matrices the merger graph used in Section~\ref{section:merging} is in fact a path. In this case, at most one subtree is split in the execution of {\sc increase}. Therefore we get the stronger guarantee that $\IP \leq \LP + c_{\max}$. Plugging in this into the proof of Theorem~\ref{theorem:approx} gives the desired result.
\end{proof}

\begin{corollary} There is a 2-approximation for Partial Rectangle Stabbing and a $d$-approximation for Partial $d$-dimensional Rectangle Stabbing.
\end{corollary}

}{}

\section{Concluding remarks and open problems}

The results in this paper suggest that Lagrangian relaxation is a powerful technique for designing approximation algorithms for partial covering problems, even though the black-box approach may not be able to fully realize its potential.

It would be interesting to extend this study on the strengths and limitation of Lagrangian relaxation to other problems. The obvious candidate is the $k$-Median problem. \citet{JV01} designed a $2\alpha$-approximation for $k$-Median using as a black box an $\alpha$-LMP approximation for Facility Location. Later, \citet{JMMSV03} gave a 2-LMP approximation for Facility Location. Is the algorithm in \cite{JV01} optimal in the sense of Theorem~\ref{theorem:lowerbound}? Can the algorithm in \cite{JMMSV03} be turned into a 2-approximation for $k$-Median by exploiting structural similarities when combining the two solutions?

\ifthenelse{\boolean{journal}}{The class of totally unimodular matrices is undoubtedly the most important subclass of balanced matrices. An open problem is to establish good approximations for Partial Totally Unimodular Cover (P-TUC). The matrix used in Section~\ref{section:gap-example} is also totally unimodular, so the lower bound on the integrality gap applies for P-TUC as well. Does the upper bound of Theorem~\ref{theorem:LP-gap} also hold P-TUC?}{}

\vspace{0.5cm}
\noindent{\bf Acknowledgments:} I am indebted to Danny Segev for sharing an early draft of \cite{KPS06} and for pointing out Kolen's work. Also thanks to Mohit Singh and Arie Tamir for helpful discussions and to Elena Zotenko for suggesting deriving the result of Section~\ref{section:lowerbound}.

\small

\bibliographystyle{abbrvnat}
\ifthenelse{\boolean{journal}}{\bibliography{references,conferences}}{\bibliography{references,conferences-short}}

\ifthenelse{\boolean{journal}}
{
\normalsize
\appendix

\section{Partial Totally Unimodular Cover} \label{appendix:lowerbound-TU}

\begin{theorem} \label{theorem:lowerbound-TU} Partial totally unimodular cover cannot be approximated better than $\frac{4}{3}$ using Lagrangian relaxation and a 1-LMP algorithm $\myalgo{A}$ as a black box.
\end{theorem}

\begin{proof}
The instance is similar to that used in Theorem~\ref{theorem:lowerbound}: The $A$-sets and the $B$-sets are the same; for each $i$ we define $O_i$ as $B_i$ minus the rightmost element. The cost of each $A$, $B$ and $O$ set is $\frac{2}{3}$, $\frac{4}{3}$ and 1 respectively. The target coverage parameter is again $P= q (q^2 +1)$.

It is straightforward to check that Lemmas~\ref{lemma:structure-prize-collecting} and~\ref{lemma:naughty-LMP} still holds for our new instance and $\alpha=1$. Then the same argument used in the proof of Theorem~\ref{theorem:lowerbound} gives us a lower bound of $\frac{4}{3}$.

It only remains to show that the resulting element-set incidence matrix $A$ is totally unimodular. A matrix $A$ is totally unimodular if and only if every submatrix $A'$ of $A$ has an equitable coloring \cite{G62}. An equitable coloring of a 0,1 matrix $A'$ is a partition of its columns into red and blue columns such that in every row of $A'$ the number of blue 1's and red 1's differs by at most one. Let us construct an equitable coloring for $A'$: all the $A$-sets are colored blue; for each $i$, if $B_i$ and $O_i$ are present in $A'$ then color one red and the other blue, and if only one is present then color it red. Clearly the coloring is equitable; thus, $A$ is totally unimodular.
\end{proof}

\section{Finding a threshold value} \label{section:thresholdvalue}

The idea is to use parametric search treating $\lambda$ as an unknown which lies in a certain range $(\lambda_l, \lambda_r)$. Initially $\lambda_l=0$ and $\lambda_r=\max_{i,j} \frac{c_j}{p_i}$. Residual capacities and dual variables are kept as a linear function of $\lambda$. We maintain the invariant that $\lambda_l^+$ separates less than $P$ profit and $\lambda_r^-$ covers more than $P$ profit. Suppose that in the interval $(\lambda_l, \lambda_r)$ the algorithm \emph{agrees} on the first $i$ elements. By this we mean that if we run the algorithm with any value $\lambda \in (\lambda_l, \lambda_r)$ the value of the dual variables of these $i$ elements (as a function of $\lambda$) is always the same. In each iteration we either find a threshold value or we narrow the interval such that the algorithm agrees on one more element, while maintaining the invariant. This cannot go on forever because the algorithm will eventually behave the same throughout the interval and the invariant would be violated. If at some point along the way we find a value of $\lambda$ covering exactly $P$ profit we stop as the solution is optimal. For simplicity, from now on we assume that this never happens.

\begin{figure}[h]

\psset{unit=1.25cm}
\begin{center}
   \begin{pspicture}(0,0.5)(6,3.5)
   \psline[arrows=->](1,1)(1,2.7)
   \psline[arrows=->](1,1)(5,1)
   \psline(1.5,1.4)(3.4,2.7)
   \psline(1.6,1.8)(4.2,2.5)
   \psline(2,2.4)(4.4,1.8)
   \psline(2.8,2.8)(4.4,1.4)
   \psline[linestyle=dotted](1.7,1.5)(1.7,1)
   \psdots[dotstyle=o](1.7,1)
   \rput(1.7,0.5){$\lambda_l$}
   \psline[linestyle=dotted](2.4,2)(2.4,1)
   \psdots[dotstyle=o](2.4,1)
   \rput(2.4,0.5){$\lambda_1$}
   \psline[linestyle=dotted](2.95,2.1)(2.95,1)
   \psdots[dotstyle=o](2.95,1)
   \rput(2.95,0.5){$\lambda_2$}
   \psline[linestyle=dotted](3.77,1.95)(3.77,1)
   \psdots[dotstyle=o](3.77,1)
   \rput(3.77,0.5){$\lambda_3$}
   \psline[linestyle=dotted](4.25,1.5)(4.25,1)
   \psdots[dotstyle=o](4.25,1)
   \rput(4.25,0.5){$\lambda_r$}
   \rput(5.4,0.8){$\lambda$}
   \rput(0.8,3.2){$c'(\cdot)$}
 \end{pspicture}
 \end{center}

 \caption{Narrowing the interval for $\lambda$. \label{figure:narrowing-interval}}
\end{figure}
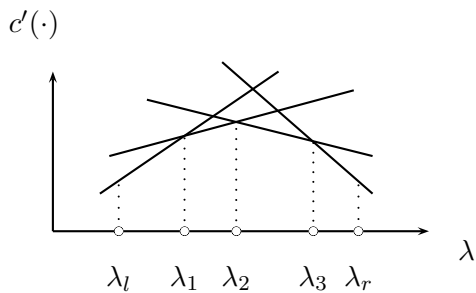

Suppose that {\scshape kolen} agrees on the first $i-1$ elements in the interval $(\lambda_l, \lambda_r)$. Note that the residual costs and $i$'s penalty are linear functions of $\lambda$. As a result, which set has the minimum residual cost, and thus which one becomes tight, if any, varies with $\lambda$. Our goal is to narrow the interval such that the set that becomes tight is always the same, or $y_i = p_i \lambda$ within the new interval. If we draw the lines corresponding to the residual costs of set that $i$ belongs to and $p_i \lambda$, the segments on the lower envelope correspond to the next tight event, either a set or element $i$. See Figure~\ref{figure:narrowing-interval}. Let $\lambda_1, \ldots, \lambda_s$ correspond to the intersection points of the lower envelope, and let $\lambda_0 = \lambda_l$ and $\lambda_{s+1} = \lambda_r$. For every $0 \leq a \leq s$, within the interval $(\lambda_a, \lambda_{a+1})$ the algorithm agrees on $i$. Note that either one of the $\lambda_a$ is a threshold value, or there exists an $a$ such that $\lambda_a^+$ covers less than $P$ profit and $ \lambda_{a+1}^-$ covers more than $P$ profit. Given the latter we update $\lambda_l = \lambda_a$ and $\lambda_r = \lambda_{a+1}$ and repeat.

\begin{theorem} A threshold value can be found by making $O\big(|U| \log |\collection{S}|\big)$ calls to {\sc kolen}.
\end{theorem}

\begin{proof}
The above discussion outlines the algorithm for finding a threshold value. Regarding the time complexity, when searching for the next tight event, instead of trying every $\lambda_a$ value, we can use binary search. Thus, only $\log |\collection{S}| + 1$ calls to { \scshape kolen} are needed to find the right value of $a$ to narrow the interval.
\end{proof}

}{}

\end{document}